\documentclass[11pt, a4paper]{article}
\usepackage[utf8]{inputenc}
\usepackage[left=2cm,right=2cm,top=3cm,bottom=3cm, a4paper]{geometry}
\usepackage{amsmath, amssymb, amsthm, enumerate}
\usepackage[T2A]{fontenc}
\usepackage{tikz}
\usepackage{comment}
\usepackage{pgfplots}
\usepackage{graphicx}
\usepackage{xcolor}
\usepackage{tabularx}
\usepackage[english]{babel}
\usepackage{algorithm}
\usepackage{algpseudocode}
\theoremstyle{plain}
\newtheorem{Th}{Theorem}[section]

\newtheorem{Prop}[Th]{Proposition}
\newtheorem{Cor}[Th]{Corollary}
\usepackage{listings}
\theoremstyle{definition}
\newtheorem{Def}{Definition}[section]

\newtheorem{Ex}{Example}[section]

\newcommand{\cS}{{\mathfrak S}}
\newcommand{\cX}{{\mathfrak X}}

\newcommand{\cF}{{\mathfrak F}}

\newcommand{\norm}[1]{\left\lVert#1\right\rVert}

\newtheorem{rem}{Remark}[section]

\newcommand{\cL}{{\cal L}}

\newcommand{\bgeqn}{\begin{eqnarray}}
\newcommand{\edeqn}{\end{eqnarray}}
\newcommand{\bgeq}{\begin{eqnarray*}}
\newcommand{\edeq}{\end{eqnarray*}}
\newcommand{\bec}{\begin{center}}
\newcommand{\enc}{\end{center}}

\newcommand{\var}{{\rm Var}}

\newcommand{\half}{ \mbox{\small$\frac{1}{2}$}}

\newcommand{\be}{\begin{equation}}
\newcommand{\ee}{\end{equation}}

\def\ess {{\rm ess\, sup}}
\def\esi {{\rm ess\, inf}}

\def\cvar{{\sf CVaR}}

\def\bbr{{\Bbb{R}}} %real numbers
\def\bbe{{\Bbb{E}}} %expectation

\def\bbp{{\Bbb{P}}}

\def\ebr{\overline{\Bbb{R}}}
\usepackage [square,sectionbib]{natbib}

\def\text#1{\;\,\hbox{#1}\;\,}    

\def\lset{\big\{\,}    \def\mset{\,\big|\,}   \def\rset{\,\big\}}
       
\outer\def\proclaim #1. #2
   \par{\medbreak\noindent{\bf#1.\enspace}{\sl#2}\par
  \ifdim\lastskip<\medskipamount \removelastskip\penalty55\medskip\fi}

\def\eop{\hfill{$\vcenter{\hrule height1pt \hbox{\vrule width1pt height5pt
   \kern5pt \vrule width1pt} \hrule height1pt}$} \medskip}
\def\low#1{{\lower1pt \hbox{$\scriptstyle #1$}}}
\def\high#1{{\raise1pt \hbox{$\scriptstyle #1$}}}

\def\argmin{\mathop{\rm argmin}}

\def\half{{{}\raise 1pt \hbox{$\frac{\scriptstyle 1}{\scriptstyle 2}$}}}
\def\eqalign#1{\begin{array}{lcr} #1 \end{array}}
\def\reals{{I\kern-.35em R}} \def\mdot{{\kern-.02em\cdot\kern-.04em}}
\def\var{{\rm VaR}} \def\cvar{{\rm CVaR}} 
  
 \def\newpage{\vfill\eject}

  \def\cC{{\cal C}} \def\cD{{\cal D}} 
\def\cE{{\cal E}} \def\cF{{\cal F}}   
  \def\cL{{\cal L}}  
  \def\cR{{\cal R}} \def\cS{{\cal S}} 
\def\cV{{\cal V}}  \def\cX{{\cal X}} 

\def\cvar{{\rm CVaR}}

\def\<x>{\langle\!\langle\mathbf{x}\rangle\!\rangle}
\def\l<{\langle\!\langle}
\def\r>{\rangle\!\rangle}

\usepackage{hyperref} % remove borders around links in pdf
\usepackage{xurl}
\hypersetup{
    colorlinks=true,
    linkcolor=blue,
    filecolor=magenta,      
    urlcolor=cyan,
    pdftitle={Overleaf Example},
    pdfpagemode=FullScreen,
    }
\usepackage{booktabs}
\pgfplotsset{compat=1.17}

\usepackage{subfigure}
\usepackage{wrapfig}
\usepackage{float}
\usepackage{multirow}
\usepackage{multicol}
\title{Biased Mean Quadrangle and Applications}
\author{
Anton Malandii\thanks{Department of Applied Mathematics and Statistics, Stony Brook University, Stony Brook, NY 11794, USA; Division of Applied Mathematics, Brown University, Providence, RI 02912, USA;\\
emails: \url{anton.malandii@stonybrook.edu}, \url{anton_malandii@brown.edu}}
\and
Stan Uryasev\thanks{Department of Applied Mathematics and Statistics, Stony Brook University, Stony Brook, NY 11794, USA;\\
email: \url{stanislav.uryasev@stonybrook.edu}}
}
\begin{document}

\maketitle
\begin{abstract}
\noindent

This paper introduces \emph{biased mean regression}, estimating the \emph{biased mean}, i.e.,  $\bbe[Y] + x $, where $ x \in \bbr$. The approach addresses a fundamental statistical problem that covers numerous applications. For instance, estimation of factors driving portfolio loss exceeding expected loss by a specified amount (e.g., $x=$ \$10 billion) or estimation of factors impacting a specific excess release of radiation in the environment (nuclear safety regulations specify different severity levels). 

The estimation is done by minimizing the so-called \emph{superexpectation error}. We establish two equivalence results that connect the method to popular paradigms: (i) biased mean regression is equivalent to quantile regression for an appropriate parameterization and is equivalent to the ordinary least squares for $x = 0$; (ii) in portfolio optimization, minimizing \emph{superexpectation risk}, associated with the superexpectation error, is equivalent to CVaR optimization. The approach is computationally attractive, as minimizing the superexpectation error is reduced to linear programming (LP), offering algorithmic and modeling advantages. It is 
 a good alternative to the ordinary least squares (OLS) regression.

The approach is based on the \emph{Risk Quadrangle} (RQ) framework, which links four stochastic functionals—error, regret, risk, and deviation—through a statistic. For the biased mean quadrangle, the statistic is the biased mean. We study properties of the new quadrangle, such as \emph{subregularity}, and establish its relationship to the quantile quadrangle. Numerical experiments confirm theoretical statements and illustrate the practical implications.
\end{abstract}
\section{Introduction}
\paragraph{Motivation.} 
One of the fundamental goals in statistical modeling is to identify covariates driving significant
excesses from baseline behavior---that is, to explain overperformance or underperformance relative to
$\mathbb{E}[Y]$ by a user-chosen margin $x$ in the units of $Y$.
Let $Y$ denote the response (regressant) and $\mathbf{X}=(X_1,\ldots,X_d)$ the covariates (regressors, factors). The question we target is:
$$
\textit{Which factors in }\mathbf{X}\textit{ explain why }Y\textit{ exceeds (or falls short of) its mean } \mathbb{E}[Y]\textit{ by a fixed margin }x>0?
$$
This problem recurs across various domains. Here are some examples.
\begin{itemize}
\item \emph{Finance:} with $Y$ a portfolio loss and its driving factors $\mathbf{X}$, which factors are responsible for losses exceeding the expected loss by $x = \$ 10$ billion? 
\item \emph{Material design:} Let $Y$ denote a target performance metric for a candidate material (e.g., the
ultimate tensile strength (UTS) measured in $\mathrm{MPa}$) and let $\mathbf{X}$
collect design and process features such as composition (wt.\%), heat-treatment
temperature ($^\circ\mathrm{C}$), holding time ($\mathrm{h}$), cooling rate ($^\circ\mathrm{C}/\mathrm{s}$),
and microstructure descriptors (grain size $\mu\mathrm{m}$, porosity vol.\%). Which settings of $\mathbf{X}$ are associated with producing parts that underperform the baseline
UTS and $x = 1 \,\mathrm{MPa}$ below the population mean? \item \emph{Medicine:} with $Y$ a recovery metric and $\mathbf{X}$ treatment features, which factors are associated with patients beating the average recovery by a clinically meaningful margin? \item \emph{Supply chain and operations:} with $Y$ delivery time and $\mathbf{X}$ capturing supplier reliability and transportation costs, which factors explain deliveries better than average by $x = 2$ hours?
\end{itemize}

A standard tool for biased estimates is \emph{quantile regression} \citep{KoenkerBassett}, which models conditional quantiles indexed by a confidence level $\alpha\in(0,1)$. While powerful, quantile regression does not specify how far above (or below) the mean the estimate should be (by design, the quantile regression estimates quantiles).
By contrast, the \emph{biased mean regression} (developed in this paper) specifies the margin of practical interest $x$ directly in the units of $Y$. The bias is specified relative to the expectation $\mathbb{E}[Y]$. This yields estimates that are easy to understand (``factors leading to under/overperformance the average by exactly $x$ units'') and easy to align with policy or design constraints.

Importantly, the biased-mean and quantile regression are not in conflict. We show that for each bias $x$ there exists a corresponding quantile level $\alpha$ that produces the same estimate, and conversely (under mild conditions). Thus, the biased mean approach can be seen as a \emph{reparameterization} of quantile regression: the user selects a margin $x$ in target units rather than a probability level $\alpha$. This reparameterization offers an alternative modeling perspective, and through the RQ integrates with risk and deviation measures used in portfolio optimization, management, and control.

\paragraph{Risk Quadrangle Basics.} 
Risk management, statistical estimation, and stochastic programming are often treated as distinct disciplines; yet, many of their core concepts --- random losses, objective functionals, and constraints --- share the same mathematical structure and properties. The \emph{Risk Quadrangle} (RQ)  \citep{Quadrangle} provides a unified framework for combining four stochastic functionals--error, regret, risk, and deviation linked by a functional, called statistic. 
Further are some important theoretical and practical implications of the RQ framework. 
\begin{enumerate}
    \item Construction of efficient numerical algorithms for optimization of elements of the quadrangles, in particular risk measures \citep{CVaR, RiskTuning,rockafellar2014random,rockafellar2015residual_risk,Anton2024expectile, cvaropt}. 

    \item Various equivalence statements related to optimization of elements of quadrangles -- for example, the equivalence between expectile and CVaR portfolio optimization \citep{Anton2024expectile}, between $\nu$-SVR and $\varepsilon$-SVR \citep{Anton2022SVR}, and between CVaR regressions based on CVaR2 and Rockafellar errors \citep{golodnikov2019cvar}.  For a detailed review on risk-adaptive approaches in stochastic programming, see \citep{Royset2022}.
    
    \item The relation between the error and deviation functionals (see Definition \ref{risk quadrangle}) implies the error-shaping decomposition of regression \citep{RiskTuning}. This decomposition establishes the equivalence between error minimization and deviation minimization with a constraint on the associated statistic, and is a part of the general \cite[Regression Theorem]{Quadrangle}. 

    \item Furthermore, the \cite[Regression Theorem]{Quadrangle} enables conditional estimation of various distributional characteristics, i.e., statistics (e.g., quantiles, expectiles) and their mixtures via regression. This is achieved by constructing a quadrangle for the statistic of interest. 
\end{enumerate}

\paragraph{Contributions.} 
This paper introduces a new quadrangle --- the \emph{biased mean quadrangle} --- and investigates its mathematical properties and practical implications. Its statistic is the mean plus a real parameter (\emph{bias}).
The associated risk measure --- the \emph{superexpectation risk} ---generalizes the well-known mean-absolute risk \citep{stochproglec}. The corresponding error, termed the \emph{superexpectation error}, is a stochastic compound functional that balances the averages of the positive and negative parts of the random variable. As in other quadrangles, the regret and deviation components are obtained from the risk and error via the standard RQ identities (see Definition \ref{risk quadrangle}).

In particular, we emphasize the following results.
\begin{enumerate}

 \item \emph{Reparametrization of quantile regression.}
We show that biased-mean regression using the superexpectation error is a reparametrization of quantile regression. In particular, quantiles can be represented as biased means with a prescribed bias $x \in \mathbb{R}$. We prove the equivalence of these two regressions; see Theorem~\ref{Th: equivalence of qr and bmr}.

 \item \emph{OLS as an LP.} When the bias parameter is zero ($x=0$), the resulting (piecewise-linear) mean quadrangle is regular--and, in fact, coherent (see Corollary~\ref{cor:reg+coherence}). This quadrangle serves as a coherent alternative to the standard mean quadrangle \citep[Example~$1'$]{Quadrangle}. In particular, we show the equivalence between OLS and the regression induced by the superexpectation error of this new mean quadrangle (piecewise-linear version), which we term the \emph{superexpectation regression} (see Theorem~\ref{th: equiv of ols and bmr}). Moreover, the superexpectation regression admits an LP formulation, offering computational and modeling advantages (See Subsection~\ref{sec: Sparse Regression} and Subsection~\ref{subsec:sparse_numerical}).

  \item \emph{Equivalence of deviation minimization.} We prove the equivalence between CVaR deviation and superexpectation deviation minimization (see Proposition~\ref{prop:equivalence}) and showcase its practical significance in the case of portfolio optimization (see Example~\ref{portfolio opt example} and Subsection~\ref{subsec:portfolio_equival}).   
 
    \item \emph{Subregularity.} We prove that the biased mean quadrangle is subregular (see Definition~\ref{def:subregquadrangle} and Proposition~\ref{Regularity of BMQ}). The notion of subregularity, first proposed in~\cite{Rockafellar2024,Quadrangle2}, extends the concept of regularity introduced in~\cite{Quadrangle} and can be viewed as an alternative to the well-known notion of coherence~\citep{artzner1999coherent}. The biased mean quadrangle also provides, to our knowledge, the first example of a subregular quadrangle with concrete practical relevance.

\end{enumerate}
\paragraph{Outline.} Section~\ref{sec:Math Background} provides definitions and theorems needed for the subsequent sections. Section~\ref{sec:BMQ} introduces the biased mean quadrangle and examines its properties --- including subregularity, its relationship to the quantile quadrangle, and related implications. Section~\ref{sec: Gen Regression} formulates and analyzes generalized regression problems within the biased mean quadrangle framework. In particular, it establishes the equivalence between quantile regression and biased mean regression. We then specialize to the zero-bias case and show equivalence between OLS and superexpectation regression (i.e., biased mean regression with zero bias). The section concludes with a discussion of sparse regression. Section~\ref{sec: Case Studies} conducts several numerical case studies showcasing the effectiveness and flexibility of the framework. Section~\ref{sec:Conclusion} concludes the paper.

\section{Mathematical Background}\label{sec:Math Background}
This section provides the definitions and theorems necessary for the subsequent analysis. It opens with definitions of value-at-risk, conditional value-at-risk, and the components of the RQ, and closes with two key theorems that are central to the proofs of the main results.  
\subsection{VaR, CVaR, and Risk Quadrangle}
 Let $(\Omega, \mathcal{A}, \bbp)$ be a probability space. For $p \in [1,\infty]$, let ${\cal L}^p(\Omega):={\cal L}^p(\Omega, {\cal A},\bbp)$ be a normed space of all random variables $X$ with $\|X\|_p<\infty$, where $\|X\|_p=(\bbe[|X|^p])^{1/p}$ for $p<\infty$, and $\|X\|_\infty=\mathrm{ess}\sup|X|$.
 %for which we also assume that $\|X\|_1 < \infty$ or $\|X\|_1 \to \infty$ whenever necessary. 
 The cumulative distribution function is denoted by $ F_X(x) = \bbp(X\leq x)$. 
 %The mathematical expectation of a random variable $X$ is denoted by $\bbe[X]$. 
 The positive and negative part of a number $x \in \bbr$ is $x_+ = \max \{ 0, x\}$ and $x_- = \max\{0,-x\}$, respectively. 
 
 Let $\overline{\bbr} := \bbr \cup \{+\infty\}$ denote an extended set of real numbers. Then, a functional $\cF: \cL^p(\Omega) \to \overline{\bbr}$ is called \emph{convex} if 
 $$\cF\left(\lambda X + (1-\lambda)Y\right) \leq \lambda \cF(X) + (1-\lambda)\cF(Y), \ \forall \; X,Y\in \cL^p(\Omega), \ \lambda \in [0,1],
 $$
 and \emph{closed (or lower-semicontinuous)} if 
 $$\left\{ X \in \cL^p(\Omega) \mset \cF(X) \leq c\right\} \ \textrm{is a closed set in } \cL^p(\Omega) \ \forall \; c < \infty.
 $$
\begin{Def}[VaR]\label{var}
The \textit{value-at-risk} (VaR, quantile) of a random variable X at a confidence level $\alpha \in [0,1]$ is defined by
\begin{equation}
    \var_\alpha(X) := [\var_\alpha^- (X), \var_\alpha^+ (X)],
\end{equation}
where
\begin{equation*}
    \var^-_\alpha(X) := \begin{cases}
 \sup \left\{ x\mset F_X(x) < \alpha\right\}, &\alpha \in (0,1] \\
 \esi\; X, &\alpha = 0  \;\;
 
\end{cases}
\end{equation*}

\begin{equation*}\label{quantile+}
    \var_\alpha^+ (X) := \begin{cases}
 \inf \left\{ x\mset F_X(x) > \alpha\right\},&\alpha \in [0,1) \\
 \ess\; X ,&\alpha = 1    \;\;
 
\end{cases}
\end{equation*}
If $\var_\alpha^-(X) = \var_\alpha^+(X)$, then 
$
    \var_\alpha(X) = \var_\alpha^-(X) = \var_\alpha^+(X).
$
\end{Def}

\begin{Def}[CVaR]\label{cvar}
The \textit{conditional value-at-risk} (CVaR) (also called AVaR,  superquantile, expected shortfall) of a random variable $X$ at confidence level $\alpha \in [0,1]$ is defined by 
\begin{equation}
    \cvar_{\alpha}(X) := \dfrac{1}{1-\alpha}\displaystyle \int\limits_{\alpha}^{1}\var^-_{\beta}(X) \, d\beta, \quad \alpha \in (0,1).
\end{equation}
For $\alpha = 0:$ 
\begin{equation*}
    \cvar_0(X) := \bbe[X] \, . 
\end{equation*}
For $\alpha = 1:$
\begin{equation*}
    \cvar_1(X) := \lim_{\tau \to 1} \cvar_{\tau}(X) =  \ess \;X\, . 
\end{equation*}
\end{Def}
\cite{Quadrangle} introduced a novel concept known as the Risk Quadrangle (RQ), aiming to establish a connection between risk management, reliability, statistics, and stochastic optimization theories. The RQ methodology groups functions of a random value $X$ into quadrangles, each comprising four elements: 1) \emph{risk,} providing a numerical surrogate for the overall hazard in $X$, 2) \emph{deviation,} measuring the ``nonconstancy'' in $X$, 3) \emph{regret,} which is similar to the utility function with a negative sign, and 4) \emph{error,} quantifying the ``nonzeroness'' in $X$.   Elements of the quadrangle are ``linked'' by a \emph{statistic} function.
\medskip
$$\eqalign{
   \hskip20pt \text{risk} \cR \,\longleftrightarrow \,\cD \text{deviation} &\cr
   \hskip50pt  \uparrow \hskip09pt \cS \hskip09pt \uparrow &\cr
   \hskip12pt \text{regret} \cV \,\longleftrightarrow \;\cE \text{error} &\cr
}$$
\smallskip

\centerline{ \bf Diagram~1:\quad The Risk Quadrangle\hskip15pt }
\medskip
\begin{Def}[Regular risk]\label{regular risk measure}
A closed convex functional $\cR: \cL^p(\Omega) \to \overline{\bbr} $ is called a \textit{regular measure of risk} if it satisfies:
 $$(R1) \,\ \cR(C) = C, \quad \forall \; C = const \qquad  \text{and} \qquad (R2)\, \ \;\cR(X) > \bbe[X], \quad \forall \; X \neq const\; .$$ 
\end{Def}
\begin{Def}[Coherent risk] A closed convex functional $\cR: \cL^p \to  \overline{\bbr}$ is called a \emph{coherent measure of risk in basic sense} if
\begin{itemize}
    \item[(C0)] $\cR(C) = C  \textrm{ for constants } C;  
$
\item[(C1)] $\cR(\lambda X) = \lambda \cR(X)$ for all $\lambda >0;$  
\item[(C2)] $\cR(X) \leq \cR(Y)$ for all $X$ and $Y$ such that $X \leq Y$ almost surely, 
\end{itemize}
and \emph{coherent in  general sense} if (C1) is left out.
\end{Def}
\begin{Def}[Regular deviation]\label{regular deviation measure}
A closed convex functional $\cD: \cL^p(\Omega) \to \overline{\bbr}^+$ is called a \textit{regular measure of deviation} if it satisfies:

    $$(D1) \,\  \cD(C) = 0, \quad \forall \; C = const \qquad  \text{and} \qquad 
    (D2) \,\ \cD(X) > 0, \quad \forall \; X \neq const.$$

\end{Def}

\begin{Def}[Regular regret]\label{regular regrt measure}
A closed convex functional $\cV: \cL^p(\Omega) \to \overline{\bbr}$ is called a \textit{regular measure of regret} if it satisfies the following axioms

    $$(V1) \,\ \cV(0) = 0 \qquad  \text{and} \qquad 
    (V2)\, \ \cV(X) > \bbe[X], \quad \forall \; X \neq 0\;.$$

\end{Def}

\begin{Def}[Regular error]\label{reg error measure}
A functional $\cE: \cL^p(\Omega) \to \overline{\bbr}^+ $ is called a \textit{regular measure of error} if it satisfies the following axioms:
    $$(E1) \,\ \cE(0) = 0 \qquad \text{end} \qquad 
    (E2)\,\ \cE(X) > 0, \quad \forall \; X \neq 0\;.$$
\end{Def}

\begin{Def}[Regular risk quadrangle] \label{risk quadrangle} A quartet $(\cR,\cD,\cV,\cE)$ of regular measures of risk, deviation, regret, and error is called a \emph{regular risk quadrangle with statistic} $\cS$ if it satisfies the following relationship formulae:
\begin{itemize}
    \item[(Q1)] \textbf{error projection:} $\cD(X)= \min\limits_C\!\lset\cE(X-C)\rset$,
    \item[(Q2)] \textbf{regret formula:} $\cR(X)= \min\limits_C\!\lset C+\cV(X-C)\rset$,
    \item[(Q3)] \textbf{centerness:} $\cR(X) = \cD(X) + \bbe[X], \quad \cV(X) = \cE(X) + \bbe[X],$
\end{itemize}
% and the quartet $(\cR,\cD,\cV,\cE)$ is ``linked'' by the statistic $\cS(X)$ that satisfies: 
\begin{itemize}
    \item[(Q4)] \textbf{statistic:} $ \cS(X) =\argmin_C\!\lset\cE(X-C)\rset
          =\argmin_C\!\lset C+\cV(X-C)\rset. $
\end{itemize}
\end{Def}

\begin{Def}[Subregular risk quadrangle]\label{def:subregquadrangle}
A quartet $(\cR,\cD,\cV,\cE)$ is called a \emph{subregular risk quadrangle} if
measures of risk, deviation, regret, and error comprising it are regular with non-strict inequalities in axioms (R2), (D2), (V2), (E2) for any $X$, and additionally there exists $\lambda > 0$ such that $\cR(\lambda X) > \bbe[\lambda X], \  \cV(\lambda X) > \bbe[\lambda X]$ for non-constant $X$ and $  \cD(\lambda X) > 0, \  \cE(\lambda X) > 0$ for non-zero $X.$
\end{Def}

\begin{Def}[Coherent risk quadrangle]
 A quartet $(\cR,\cD,\cV,\cE)$ satisfying relationship formulae (Q1)--(Q4), where $\cR$ is a coherent measure of risk is called a \emph{coherent risk quadrangle.} 
\end{Def}
\begin{rem}[Error Projection and Regret Formula]\label{remark proj cert} To prove that a given quartet $(\cR, \cD, \cV, \cE)$ is a quadrangle, it is sufficient to verify either conditions (Q1) and (Q3), or conditions (Q2) and (Q3), as conditions (Q1) and (Q2) are intrinsically linked through the condition (Q3). Indeed,
\begin{equation*}
    \cR(X) = \min\limits_C\!\lset C+\cV(X-C)\rset = \min\limits_C\!\lset\cE(X-C)\rset + \bbe[X] = \cD(X) + \bbe[X].
\end{equation*}
Consequently, (Q4) holds automatically. 
\end{rem}
The following example represents a well-known regular risk quadrangle that has gained widespread utilization in optimization, statistics, and risk management.
\begin{Ex}[Quantile Quadrangle] 
The quantile quadrangle (see Diagram 2) is named after the VaR (quantile) statistic. This quadrangle, as introduced by \cite{Quadrangle}, establishes a relationship between the CVaR optimization technique discussed in their previous works \citep{CVaR,CVaR2} and quantile regression \citep{KoenkerBassett, KoenkerBook}.
    \medskip
$$\eqalign{
  \hskip40pt
  \cS_\alpha (X)=\var_\alpha(X)=\!\!\text{VaR}
&\cr
  \hskip40pt
  \cR_\alpha  (X)  =\cvar_\alpha(X) =\!\!\text{CVaR} &\cr
  \hskip40pt
  \cD_\alpha  (X)=\cvar_\alpha(X) - \bbe[X]
                    =\!\!\text{CVaR deviation} &\cr
  \hskip40pt
  \cV_\alpha  (X) =\frac{1}{1-\alpha} \bbe [X_+]
                    =\!\!\text{scaled partial moment} &\cr
  \hskip40pt
  \cE_\alpha  (X) = \bbe [ \frac{\alpha}{1-\alpha} X_+ +X_-]
                   =\!\!\text{normalized Koenker--Basset error} &}
$$
\smallskip

\centerline{{\bf Diagram~2: \quad Quantile Quadrangle}
                (at confidence level $\alpha\in (0,1)$) \label{diag 2}}

\medskip
\end{Ex}
\begin{rem}[Notation] In the following sections, the distinction between elements of different quadrangles (e.g., different risk measures or error measures) is achieved by assigning a subscript parameter to each element of a quadrangle. For example, $\cR_\alpha(X)$  stands for $\cvar_\alpha(X),$ $\cR_x$  stands for the superexpectation risk (see Diagram 3), and similarly for other quadrangle corners. To avoid confusion, we never explicitly write expressions like $\cR_{0.875}(X).$ Instead, we write $\cR_\alpha(X)$ for $\alpha = 0.875.$ The only exception is $\cR_0(X),$ which is $\cR_x(X)$ for $x=0$. 
    
\end{rem}

\subsection{Superexpectation and Dual CVaR Optimization Formula}
The popularity of the CVaR risk measure can mostly be attributed to two essential properties. Firstly, its interpretability as the average of the $(1-\alpha) \times 100 \ \%$ worst-case outcomes within the distribution's tail. Secondly, its computational efficiency due to the CVaR optimization formula \citep{CVaR,CVaR2}.
\begin{Th}[CVaR Optimization Formula] \label{CVaR opt th}
For a random variable $X \in \cL^1(\Omega)$ and $\alpha \in (0,1)$, it holds
\begin{equation}\label{CVaR opt formula}
    \cvar_\alpha(X) = \min_{x \in \bbr} \left\{ x + \frac{1}{1-\alpha}\bbe [X-x]_+\right\},
\end{equation}
and the set of minimizers for (\ref{CVaR opt formula}) is $\var_\alpha(X)$.
\end{Th}
\cite{rockafellar2014random} defined the \emph{superexpectation} (SE) function (also studied by \cite{ogryczak2002dual}) as follows 
\begin{equation}\label{superexp}
    E_{\! X}(x) := \bbe [X-x]_+ + x, \quad x \in \bbr \;.
\end{equation}
Formulas 
\eqref{CVaR opt formula} and 
\eqref{superexp}
are closely related as they share the same 
partial moment function  $\bbe [X-x]_+$.
\cite{rockafellar2014random} showed that the superexpectaion is related to the CVaR through the Fenchel--Legendre transformation and proved the following Theorem \ref{dual cvar th}.
\begin{Def}[Fenchel--Legendre transform]
Let $f:\bbr \to \ebr$. Then a function defined by the following equality
\begin{equation}
    f^*(\alpha) = \sup_{x \in \operatorname{dom} f} \{ \alpha x - f(x) \},
\end{equation}
is called the \emph{conjugate} of $f$ or its \emph{Fenchel--Legendre transform}.
\end{Def}
\begin{Th}[Dual CVaR Optimization Formula]\label{dual cvar th}
For a random variable $X \in \cL^1(\Omega)$ and a scalar $x \in \bbr$, it holds
\begin{equation}\label{dual cvar formula}
    E_{\! X}(x) = \max_{\alpha \in [0,1]}\left \{\alpha x +(1-\alpha)\cvar_{\alpha}(X)\right \},
\end{equation}
and the set of maximizers for (\ref{dual cvar formula}) is $[\bbp(X<x), \bbp(X\leq x)]$. Or in terms of the Fenchel--Legendre transform:
\begin{equation}
    E_{\! X}^*(\alpha) = -(1-\alpha)\cvar_{\alpha}(X) \, .
\end{equation}
\end{Th}

\section{Biased Mean Quadrangle}\label{sec:BMQ}
This section focuses on the biased mean quadrangle (see Diagram~3) and explores its mathematical properties. It then establishes its relation to the quantile quadrangle, uncovering various potential applications. 
\medskip
$$\eqalign{
  \hskip20pt
  \cS_x (X) = x + \bbe[X] = \text{biased mean}
&\cr
  \hskip20pt
  \cR_x (X)  = E_{X-\bbe[X]}(x) - x_+ +\bbe[X]\; \;= \;\;\bbe [X-\bbe[X]-x]_+ - x_- +\bbe[X] = \!\!\text{superexpectation risk} 
&\cr
  \hskip20pt
  \cD_x (X) = E_{X-\bbe[X]}(x) - x_+\;\; = \;\;\bbe [X-\bbe[X]-x]_+ - x_- = \text{superexpectation deviation}
  &\cr
  \hskip20pt
  \cV_x  (X) = \max \{\bbe[X_-] - x_+\;,\; \bbe[ X_+] - x_-\} + \bbe[X]  = \text{superexpectation regret}
  &\cr
  \hskip20pt
  \cE_x  (X) = \max \{\bbe[X_-] - x_+\;, \;\bbe[X_+] - x_-\}=\text{superexpectation error}}
$$
\smallskip

\centerline{{\bf Diagram~3\label{diag 3}: \quad Biased Mean Quadrangle}
                ($x\in \bbr$) }

\medskip
Unless stated otherwise, assume $X \in \cL^1(\Omega)$.

\begin{Prop}[Biased Mean Quadrangle]\label{Regularity of BMQ}
The quartet $\left(\cR_x, \cD_x, \cV_x, \cE_x \right)$ is a subregular quadrangle with statstic $\cS_x$.
\end{Prop}
\begin{proof}
To prove that $\left(\cR_x, \cD_x, \cV_x, \cE_x \right)$ is a quadrangle, it suffices (see Remark \ref{remark proj cert}) to establish the validity of (Q1) and (Q3). To show that the quadrangle is subregular, it suffices to demonstrate that $\cE_x$ satisfies the Definition \ref{reg error measure} with a non-strict inequality in (E2) and show that there exists $\lambda >0$ such that $\cE_x(\lambda X) > 0$ for non-zero $X$. 

Indeed, $\cE_x$ is a closed convex functional as a maximum of closed convex functionals and
$$\cE_x(0) = \max\{-x_+,-x_-\} = 0.
$$
Hence $\cE_x(X)$ satisfies (E1). Further, since $\bbe[X_-]$ and $\bbe[X_+]$ are non-negative
$$\cE_x(X) = \max \{\bbe[X_-] - x_+, \bbe[X_+] - x_-\} \geq \max\{-x_+,-x_-\} = 0, \quad \forall \; X \neq 0.
$$
There are only two cases when $\forall \; X \neq 0, \quad \cE_x(X) = 0:$ 
\begin{itemize}
    \item[1)] $x>0, \ \bbe[X_+] = 0, \ \bbe[X_-]-x \leq0,$ hence for any $\varepsilon>0,$ $\lambda = \dfrac{x}{\bbe[X_-]} + \varepsilon > 0$ the following is valid $\cE_x(\lambda X) > 0.$
    \item[2)] $x<0, \ \bbe[X_-] = 0, \ \bbe[X_+]+x \leq0,$ hence for any $\varepsilon>0,$ $\lambda = \dfrac{-x}{\bbe[X_+]} + \varepsilon > 0$ the following is valid  $\cE_x(\lambda X) > 0.$
\end{itemize}
Thus $\cE_x$ is a subregular error. Next, let us consider the error projection formula:
\begin{equation}\label{error proj}
\min_C\{ \cE_x(X-C)\} = \min_C \max\{\bbe[X-C]_- -x_+, \bbe[X-C]_+ -x_- \}.    
\end{equation}
Note that the minimum in \eqref{error proj} is attained if and only if 
$$\bbe[X-C]_- -x_+ = \bbe[X-C]_+ -x_-
$$
or equivalently
\begin{equation}\label{statistic}
  C = x + \bbe[X] = \cS_x(X).  
\end{equation}
Therefore, by substituting \eqref{statistic} into \eqref{error proj}, we obtain 
\begin{equation*}
\min_C\{ \cE_x(X-C)\} = \bbe[X- \bbe[X] - x]_+ -x_- = E_{X-\bbe[X]}(x)-x_+ = \cD_x(X),
\end{equation*}
which completes the proof.

\end{proof}

\begin{Cor}[Mean Quadrangle (mean upper semideviation version)] \label{cor mean-upp-sem dev}
An important special case of the biased mean quadrangle is when the parameter $x=0$. In this case, the statistic equals the expected value, see Diagram~4. The deviation in this quadrangle is the mean-upper-semideviation of order 1. Therefore, we call this quadrangle the ``mean quadrangle'' (mean-upper-semideviation version). This quadrangle can be viewed as a ``linearization'' of the standard mean quadrangle \cite[Example 1$'$]{Quadrangle}. Notably, this quadrangle is coherent (see Corollary~\ref{cor:reg+coherence}).
\end{Cor}
\medskip
$$\eqalign{
  \hskip40pt
  \cS_0 (X) = \bbe[X] 
  % = \text{mean}
&\cr
  \hskip40pt
  \cR_0 (X)  = E_{X-\bbe[X]}(0) +\bbe[X]  = \bbe[X-\bbe[X]]_+ + \bbe[X] 
  % = \text{mean upper semideviation}
&\cr
  \hskip40pt
  \cD_0 (X) = E_{X-\bbe[X]}(0) = \bbe[X-\bbe[X]]_+ 
  % = \text{upper semideviation}
  &\cr
  \hskip40pt
  \cV_0  (X) = \max \{\bbe[X_-], \bbe[ X_+]\} + \bbe[X] 
  % = \text{upper semideviation regret} 
  &\cr
  \hskip40pt
  \cE_0  (X) = \max \{\bbe[X_-], \bbe[X_+] \}
  % = \text{upper semideviation error}
  }
$$
\smallskip

\centerline{{\bf Diagram~4 : \quad Mean Quadrangle (mean-upper-semideviation version)} }\label{diagram4}

\medskip

\begin{Cor}[Mean Quadrangle ($\cL^1$ norm version)]\label{cor mean-abs dev} The mean quadrangle (mean-upper-semideviation version) admits an equivalent representation in terms of $\cL^1$ norm (see Diagram~5) since (see, \cite{stochproglec})
$$\bbe[X - \bbe[X]]_+ = \bbe[X - \bbe[X]]_- = \frac{1}{2}\bbe[|X-\bbe[X]|] = \frac{1}{2}\|X- \bbe[X]\|_1 \quad \forall \; X \in \cL^1(\Omega).
$$ 
Additionally, from the fact that 
$$\max\{x,y\} = \frac{x+y}{2} + \frac{|x-y|}{2} \quad \forall \; x,y \in \bbr
$$
the expression for the error $\cE_0$ follows.
\end{Cor}
\medskip
$$\eqalign{
  \hskip40pt
  \cS_0 (X) = \bbe[X]
&\cr
  \hskip40pt
  \cR_0 (X)  =  \frac{1}{2}\|X-\bbe[X]\|_1 + \bbe[X]
&\cr
  \hskip40pt
  \cD_0 (X) =  \frac{1}{2}\|X-\bbe[X]\|_1
  &\cr
  \hskip40pt
  \cV_0  (X) = \frac{1}{2}\|X\|_1 + \frac{1}{2}|\bbe[X]| + \bbe[X] 
  &\cr
  \hskip40pt
  \cE_0  (X) = \frac{1}{2}\|X\|_1 + \frac{1}{2}|\bbe[X]|}
$$
\smallskip

\centerline{{\bf Diagram~5 : \quad Mean Quadrangle ($\cL^1$ norm version)} \label{diagram5}}

\medskip 
\begin{Cor}[Regularity of the Mean Quadrangle]\label{cor:reg+coherence}The quartet $\left(\cR_0, \cD_0, \cV_0, \cE_0 \right)$ is a regular (moreover a coherent) quadrangle with statistic $\cS_0$.  
\end{Cor}
\begin{proof}
It is obvious that $\cE_0(X)= \max\{\bbe[X_+], \bbe[X_-] \}>0, \ \forall \; X \neq 0$ and that it is convex and closed functional. Thus $\cE_0(X)$ is a regular error. The rest follows from Proposition \ref{Regularity of BMQ}.  

The proof of coherence of $\cR_0$ can be found in \citep{stochproglec}.
\end{proof}
The biased mean quadrangle has a profound relationship with the quantile quadrangle, demonstrated in Proposition~\ref{prop 2}. 

\begin{Prop}[Relation between Quantile and Biased Mean Quadrangles]\label{prop 2}
Let $\left(\cR_\alpha, \cD_\alpha, \cV_\alpha, \cE_\alpha \right)$ be a quantile quadrangle for $\alpha \in (0,1).$ Then for any $x\in \bbr$ 
\begin{align}
  \cR_x (X) &= \displaystyle \max_{\alpha \in [0,1]}\left \{(1-\alpha)(\cR_\alpha(X)-x_+)+ \alpha( \bbe[X]-x_-)\right \},\label{extended kusuoka}\\
  \cD_x (X) &= \displaystyle \max_{\alpha \in [0,1]}\left \{(1-\alpha)(\cD_\alpha(X)-x_+)- \alpha x_-\right \} ,\label{kusuoka dev}\\
  \cV_x  (X) &= \displaystyle \max_{\alpha \in [0,1]}\left \{(1-\alpha)(\cV_\alpha(X)-x_+) + \alpha( \bbe[X]-x_-)\right \},\\
  \cE_x  (X) &= \displaystyle \max_{\alpha \in [0,1]}\left \{(1-\alpha)(\cE_\alpha(X)-x_+)-\alpha x_-\right \}.\label{kusuoka error}
\end{align}

% Moreover, $\forall x\in (\inf (X), \sup (X)) \  \exists \alpha^* \in [\bbp (X < \cS_x(X)), \bbp (X \leq \cS_x(X))]:$
%\begin{equation*}
    %\min_X \left\{\cD_x(X)\right\} = \min_X \left\{  (1-\alpha^*)(\cD_{\alpha^*}(X)-x) - x_-\right\}.
%\end{equation*}

\end{Prop}
\begin{proof}
It is sufficient to prove equalities \eqref{kusuoka dev} and \eqref{kusuoka error} since the expressions for $\cR_x(X)$ and $\cV_x(X)$ follow from (Q3). Indeed, 
\begin{equation*}
    \begin{split}
      & \;\;\;\;\max_{\alpha \in [0,1]}\left \{(1-\alpha)(\cE_\alpha(X)-x_+) -\alpha x_-\right \} \\
        & = \max_{\alpha \in [0,1]}\left \{(1-\alpha)\left(\bbe \left[ \frac{\alpha}{1-\alpha} X_+ +X_-\right]-x_+\right)-\alpha x_-\right \}\\
        & = \max_{\alpha \in [0,1]}\left \{\left(\bbe \left[ \alpha X_+ +(1-\alpha)X_-\right] -(1-\alpha)x_+-\alpha x_ -\right)\right \}\\
        & = \max \{\bbe[X_-] - x_+, \bbe[X_+] - x_-\}=  \cE_x  (X) \;.
    \end{split}
\end{equation*}
To demonstrate the equality \eqref{kusuoka dev}, we use the definition of $\cD_x(X),$ i.e.,
$$\cD_x(X) = E_{X-\bbe[X]}(x) - x_+
$$
and apply the Theorem \ref{dual cvar th} to a random variable $X-\bbe[X].$ 
\end{proof}

% \cite{ShapiroKusuoka} refers to the following functional as the mean-upper-semideviation 
% $$
% \displaystyle \max_{\alpha \in [0,1]}\left \{(1-\alpha)\cR_\alpha(X) + \alpha\bbe[X]\right\}\;.
% $$
 % We refer to it as the {\it mean-upper-semirisk} since it satisfies risk axioms; see the following corollary.

\begin{Cor}[Kusuoka Representation of Mean Absolute Risk] By setting $x=0$ in \eqref{extended kusuoka}, we obtain the (minimal) Kusuoka representation (see \cite{ShapiroKusuoka}) of the mean-absolute risk, i.e.,
\begin{equation}\label{Kusuoka of sem risk}
    \cR_0(X) = \displaystyle \max_{\alpha \in [0,1]}\left \{(1-\alpha)\cR_\alpha(X) + \alpha\bbe[X]\right \}.
\end{equation}
\end{Cor}
Interestingly, it is possible to construct the mean quadrangle (mean-upper-semideviation version) starting from $\cR_0$ (note that it is generally easier to construct a quadrangle starting from either a regret or an error).
Indeed, application of Theorem \ref{CVaR opt th} to \eqref{Kusuoka of sem risk} yields
\begin{align*}
  \cR_0(X) &=  \max_{\alpha \in [0,1]} \min_{C \in \bbr}\left \{(1-\alpha)\left(C + \frac{1}{1-\alpha}\bbe[X-C]_+ \right)+ \alpha\bbe[X]\right \}\\
  &= \max_{\alpha \in [0,1]} \min_{C \in \bbr}\left \{(1-\alpha)C + \bbe[X-C]_+ + \alpha\bbe[X]\right \}\\
  &=\min_{C \in \bbr}\max_{\alpha \in [0,1]} \left \{ \bbe[X] + \bbe[X-C]_- + \alpha(\bbe[X-C]_+ - \bbe[X-C]_-)\right \}\\
  &=\min_{C \in \bbr} \ C + \max\{\bbe[X-C]_+, \bbe[X-C]_-\} + \bbe[X-C]\\
  &= \min_{C \in \bbr} \  C + \cV_0(X-C),
\end{align*}
where $\cS_0(X) = \bbe[X] = \argmin\limits_{C \in \bbr} \  C + \cV_0(X-C).$
For a more detailed discussion on $\cR_0,$ its Kusuoka representation, and statistical properties, see \cite[Chapter 6]{stochproglec}.

 The following Proposition \ref{prop:equivalence} establishes the equivalence between the CVaR deviation and the superexpectation deviation minimization. 

\begin{Prop}[Equivalence of the  CVaR and Superexpectation  Deviation Minimization]\label{prop:equivalence}
Let $\cX \subseteq \cL^1(\Omega)$ be a nonempty closed, convex, and bounded feasible set for the following optimization problems: 
\begin{equation}\label{x dev min}
    \min_{X \in \cX} \quad \cD_x(X),
\end{equation}
\begin{equation}\label{al dev min}
    \min_{X \in \cX} \quad \cD_\alpha(X).
\end{equation}
Additionally, assume that the optimal solution in \eqref{x dev min} and \eqref{al dev min} exists. 
Then $\forall \; x \in \bbr \; \exists \; \alpha \in [0,1]$ such that
\begin{equation}\label{eq div min}
    \argmin_{X \in \cX}  \  \cD_x(X)  =  \argmin_{X \in \cX}  \ \cD_{\alpha}(X). 
\end{equation}
Moreover, for $X^* \in  \argmin\limits_{X \in \cX} \ \cD_x(X), \ \alpha \in [\bbp(X^* < x + \bbe[X^*]), \ \bbp(X^* \leq x + \bbe[X^*])]$ realizes the equality \eqref{eq div min}.

\end{Prop}
\begin{proof}
By Proposition \ref{prop 2}
\begin{equation}\label{kusuoka dev in cor}
     \min_{X \in \cX} \quad \cD_x (X) =  \min_{X \in \cX}\displaystyle \max_{\alpha \in [0,1]} \quad \left \{(1-\alpha)(\cD_\alpha(X)-x_+)- \alpha x_-\right \}
\end{equation}
The left-hand side of \eqref{kusuoka dev in cor} is a convex optimization problem for which the existence of the optimal solution is assumed. 
The right-hand side of \eqref{kusuoka dev in cor} is a minimax (convex-concave) optimization problem. Since the existence of an optimal solution holds for the left-hand side of \eqref{kusuoka dev in cor}, then the function $(1-\alpha)(\cD_\alpha(X)-x_+)- \alpha x_-$ possesses a saddle point $( X^*,\alpha^*)$ on $ \cX \times [0,1].$ 
Let $X^* \in \argmin\limits_{X \in \cX} \ \cD_x(X)$  for some $x \in \bbr$. Then \eqref{kusuoka dev in cor} implies that 
$$\cD_x (X^*) = (1-\alpha^*)(\cD_{\! \alpha^*}(X^*)-x_+)- \alpha^* x_-,
$$
where $( X^*,\alpha^*)$ is a saddle point for the right-hand side of \eqref{kusuoka dev in cor}. Therefore, by the definition of a saddle point
$$
X^* \in \argmin\limits_{X \in \cX} \ \cD_{\! \alpha^*}(X).
$$

Moreover, if $X^* \in \argmin\limits_{X \in \cX} \ \cD_{\! \alpha^*}(X)$ for the same $\alpha^*,$ then $( X^*,\alpha^*)$ is a saddle point for the right-hand side of \eqref{kusuoka dev in cor} and consequently, $X^* \in \argmin\limits_{X \in \cX} \ \cD_x(X)$ for the same $x \in \bbr.$

Finally, Theorem \ref{dual cvar th} applied to a random variable $X^* - \bbe[X^*]$ implies that $\alpha^* \in [\bbp(X^* < x + \bbe[X^*]), \ \bbp(X^* \leq x + \bbe[X^*])],$ which completes the proof.
\end{proof}

\begin{rem}[Choice of the Functional Space]The choice of $\cL^1(\Omega)$ in Proposition~\ref{prop:equivalence} is natural; however, any $\cL^p(\Omega)$ with $p \ge 1$ may be used for convenience. For $p \in (1,\infty)$, the space $\cL^p(\Omega)$ is reflexive, and under the assumption that $\cD_x$ and $\cD_\alpha$ are coercive the existence of an optimal solution in~\eqref{eq div min} is guaranteed, see \cite[Chapter 7]{kurdila2005convex}.

\end{rem}
\begin{Cor}[Equivalence of SE Risk and CVaR Minimization with Expectation Constraint]\label{risk opt cor} Let $\cX \subseteq \cL^1(\Omega)$ be closed, convex, and bounded. Define the feasible set
\begin{equation*}
    \cX_\mu = \left\{X \in \cX: \bbe[X] = \mu \in \bbr \right\}.
\end{equation*}
Then   $\forall \; x\in \bbr \ \exists \; \alpha \in (0,1)$ such that
\begin{equation*}\label{risk argmins}
    \argmin_{X \in \cX_\mu} \ \cR_x(X) = \argmin_{X \in \cX_\mu} \ \cR_{\alpha}(X) \, .
\end{equation*}
\end{Cor}
\begin{proof}
Since $\cD(X) = \cR(X)-\bbe[X]$ for any (sub)regular $\cR$ and $\cD$ from the same quadrangle (see Definition~\ref{risk quadrangle}), then the corollary follows from Proposition~\ref{prop:equivalence}.      
\end{proof}
The following Example \ref{portfolio opt example} demonstrates how Corollary \ref{risk opt cor} can be applied to the portfolio optimization problem.    
\begin{Ex}[Equivalence of Portfolio Optimization]\label{portfolio opt example} Let $\cX \subseteq \cL^1(\Omega)$ be closed, convex, and bounded. Define the feasible set 
\begin{equation*}
    \cX_{\mathbf{w},\mu} = \left\{X \in \cX: X = -(w_1X_1 +\ldots + w_nX_n)\,, \quad \sum_{i=1}^n w_i = 1, \quad \bbe[-X] = \mu \right\},
\end{equation*}
where $w_i \in \bbr$ are portfolio weights, $X_i \in \cL^1(\Omega)$ are random returns of assets, $i=1,\ldots,n$, constituting a portfolio, and $\mu>0$ is a target expected return of the portfolio. Then  for every $x \in \bbr $ there exists $\alpha \in (0,1)$ such that
    $$\argmin\limits_{X \in \cX_{\mathbf{w},\mu}} \ \cR_x(X) =\argmin\limits_{X \in \cX_{\mathbf{w},\mu}} \ \cR_{\alpha}(X) \, .$$

\end{Ex}

Proposition~\ref{prop:equivalence} also plays a crucial role in statistical estimation via generalized regression, \citep{Quadrangle}. The following section discusses this in more detail.

\section{Generalized Regression within the Biased Mean Quadrangle}\label{sec: Gen Regression}

This section formulates and studies generalized regression problems in the framework of the biased mean quadrangle.

Let $Y \in \cL^p(\Omega)$ be a target random variable, i.e., a \emph{regressant} and $\mathbf{X} = (X_1,\ldots, X_d)^\top$ be a random vector of \emph{regressors} such that $X_i \in \cL^p(\Omega), \ i = 1,\ldots, d.$ In RQ framework, the \emph{generalized regression} problem is to find a function $f: \bbr^d \to \bbr,$ belonging to a class of measurable functions $\cC,$ that minimizes the \emph{regression residual} $Z_f:= Y - f(\mathbf{X})$ with respect to a particular error $\cE$ (e.g., mean square error, mean absolute error). Specifically, the goal is to solve the following stochastic optimization problem:
\begin{equation}\label{regression problem}
    \min_{f \in \cC} \quad \cE(Z_f).
\end{equation}
Fix the space of random variables to be $\cL^1(\Omega)$. By setting the superexpectation error $\cE_x, \ x \in \bbr$ as the objective in \eqref{regression problem}, we obtain the biased mean regression (BMR) 
\begin{equation}\label{bm regression problem}
    \min_{f \in \cC} \quad \cE_x(Z_f),
\end{equation}
and by setting the KB error (see Diagram 2) $\cE_\alpha, \ \alpha \in (0,1)$  as the objective in \eqref{regression problem}, we obtain the quantile regression (QR)
\begin{equation}\label{q regression problem}
    \min_{f \in \cC} \quad \cE_\alpha(Z_f).
\end{equation}
According to the \cite[Regression Theorem]{Quadrangle}, problem \eqref{bm regression problem} is equivalent to the following constrained optimization problem 
\begin{align}
    \min_{f \in \cC} &\quad \cD_x(Z_f)\label{bm regression opt}\\
    \textrm{s.t.} &\quad 0 = \cS_x(Z_f),\label{bm regression constr}
\end{align}
and analogously,   problem \eqref{q regression problem} is equivalent to
\begin{align}
    \min_{f \in \cC} &\quad \cD_\alpha(Z_f)\label{q regression opt}\\
    \textrm{s.t.} &\quad 0 \in \cS_\alpha(Z_f).\label{q regression constr}
\end{align}
Since \eqref{bm regression opt}--\eqref{bm regression constr} is a constrained deviation minimization problem, it is related to the quantile regression problem \eqref{q regression opt}--\eqref{q regression constr} via Proposition \ref{prop:equivalence}.
\begin{Th}[Equivalence of QR and BMR]\label{Th: equivalence of qr and bmr}
    Consider the regression problem \eqref{regression problem}, where the regression residual $ Z_f \in \cL^1(\Omega), \ \forall f \in \cC.$ For $x \in \bbr$  denote by $\mathfrak{S}_{se}$  a set of optimal solutions to the regression problem \eqref{bm regression problem}. Let $f^* \in \mathfrak{S}_{se},$ and denote by $\mathfrak{S}_{kb}$ a set of optimal solutions to the regression problem \eqref{q regression problem} with the KB error, where $$\alpha^* \in \mathcal{I}(Z_{f^*}) := [\bbp(Z_{f^*} <0), \ \bbp(Z_{f^*} \leq 0)] \, .$$ Then 
  \begin{equation}
      \mathfrak{S}_{se} \subseteq \mathfrak{S}_{kb}. 
  \end{equation}
\end{Th}
\begin{proof}

Let $f^* \in \mathfrak{S}_{se}$ be an optimal solution to \eqref{bm regression opt}--\eqref{bm regression constr} for some $x \in \bbr.$ Then 
\begin{equation*}
    \cD_x(Z_{f^*}) = \min_{f \in \cC} \ \cD_x(Z_f) \text{s.t.} 0 =x+\bbe[Z_{f}].
\end{equation*}
Thus by Proposition \ref{prop 2} and Theorem \ref{dual cvar th}
\begin{equation*}
    \cD_x(Z_{f^*}) = (1 -\alpha^*)(\cD_{\alpha^*}(Z_{f^*})-x_+) - \alpha^*x_- \text{with}  0=x+\bbe[Z_{f^*}] \in \var_{\alpha^*}(Z_{f^*}), \text{where} \alpha^* \in  \mathcal{I}(Z_{f^*}).
\end{equation*}
Therefore, Proposition \ref{prop:equivalence} implies that $f^*$ is an optimal solution to \eqref{q regression opt}--\eqref{q regression constr} and hence an optimal solution to \eqref{q regression problem}, i.e., $f^* \in \mathfrak{S}_{kb}.$
\end{proof}
% \begin{rem}[Reverse inclusion $\mathfrak{S}_{se} \supseteq \mathfrak{S}_{kb}$]
    
% \end{rem}
The equivalence of QR and BMR suggests an alternative way of specifying a quantile of a given distribution. Instead of setting a specific confidence level $\alpha \in [0,1],$ one can define a quantile using a parameter $x \in \bbr.$ This parameter allows a quantile to be represented as a biased (shifted) mean, where $x$ is the magnitude of the bias. The following example illustrates one of the many potential use cases of this methodology.

\begin{Ex}[The newsvendor problem]
\label{ex:newsvendor}
Let $Y\in\bbr$ denote a random demand and $\mathbf{X}\in\bbr^d$ observed contexts. Let $\gamma>0$ be the buying (unit) cost and $\delta>\gamma$ the selling (unit) price. We model the order quantity as an affine function
\[
f(\mathbf{X}) := \mathbf{c}^\top \mathbf{X} + c_0,\qquad (\mathbf{c}, c_0)\in\bbr^{d+1}\;.
\]
Define the residual $Z(\mathbf{c},c_0) := Y - f(\mathbf{X})$.

\medskip
\noindent\textbf{Question A (order policy under given prices):}
\emph{Given $(\gamma,\delta)$, what order policy $f(\mathbf{X})$ minimizes the expected underage/overage loss?}

\noindent\textbf{Answer A:}
Solve the following KB error optimization problem
\begin{equation}
\label{eq:newsvendor-pinball}
\min_{\mathbf{c},\,c_0}\;
\bbe\!\left[\left(1-\frac{\gamma}{\delta}\right)\,Z(\mathbf{c},c_0)_+ \;+\; \frac{\gamma}{\delta}\,Z(\mathbf{c},c_0)_-\right].
\end{equation}
Problem \eqref{eq:newsvendor-pinball} is the quantile regression problem at level
\[
\alpha \;=\; 1 - \frac{\gamma}{\delta}\in(0,1).
\]
Hence any minimizer $(\mathbf{c}^*,c_0^*)$ satisfies
\[
\mathbf{c}^{*\top}\mathbf{X}+c_0^* \;\in\; q_{\alpha}(Y\mid \mathbf{X})\quad\text{a.s.}
\]
In the no-context case $\mathbf{X}\equiv \mathbf{0}$, we obtain $c_0^* \in q_{\alpha}(Y)$, the classical newsvendor quantile rule.

\medskip
\noindent\textbf{Question B (pricing from a target stated in \emph{units}):}
\emph{Suppose the decision-maker specifies a practically meaningful target $x\in\bbr$ for the exceedance/shortfall of expected demand. What selling price $\delta$ makes the optimal order policy consistent with this target?}

\noindent\textbf{Answer B:}
Let $x\in\bbr$ encode a target exceedance of the mean (in the sense of the biased-mean quadrangle), and consider the superexpectation error $\cE_x$.
Solve the \emph{biased-mean regression}
\begin{equation}
\label{eq:bmr}
\min_{\mathbf{c},\,c_0}\;\cE_x \bigl(Z(\mathbf{c},c_0)\bigr) \, .
\end{equation}
Let $(\mathbf{c}^*,c_0^*)$ be any minimizer of \eqref{eq:bmr} and define the induced quantile level
\[
\alpha^* \;:=\; \bbp \bigl(Z(\mathbf{c}^*,c_0^*)\le 0\bigr) \;=\; \bbp \bigl(Y \le \mathbf{c}^{*\top}\mathbf{X}+c_0^*\bigr) \, .
\]
Then by Theorem~\ref{Th: equivalence of qr and bmr} there exists a quantile level $\alpha=\alpha^*$ such that $(\mathbf{c}^*,c_0^*)$ also solves the $\alpha$--quantile regression. Consequently, to \emph{choose the selling price} $\delta$ consistent with the biased-mean target $x$, set
\[
\delta \;=\; \frac{\gamma}{1-\alpha^*},
\]
because the newsvendor identity $\alpha = 1 - \gamma/\delta$ implies $\delta = \gamma/(1-\alpha)$ and $\alpha=\alpha^*$ at the optimum. Operationally:
\begin{enumerate}
  \item Solve \eqref{eq:bmr} for $(\mathbf{c}^*,c_0^*)$ at the desired $x$.
  \item Compute $\alpha^*=\bbp \bigl(Y \le \mathbf{c}^{*\top}\mathbf{X}+c_0^*\bigr)$.
  \item Set $\delta = \gamma/(1-\alpha^*)$.
\end{enumerate}

Formulation \eqref{eq:newsvendor-pinball} answers: \emph{``Given prices, what is the optimal order policy?''} By contrast, \eqref{eq:bmr} answers: \emph{``Given a target buffer $x$ stated in natural units of demand, what price $\delta$ realizes the resulting optimal order policy?''} The latter is often easier to elicit from practitioners who work with units rather than in confidence levels.
\end{Ex}

In particular, when $x=0,$ one recovers a quantile coinciding with (or containing) a mean of a distribution. The following Corollary~\ref{cor: unbiased l1} discusses this in more detail.

\begin{Cor}[Unbiased $\cL^1$ Regression]\label{cor: unbiased l1} For $x=0,$ problem \eqref{bm regression opt}--\eqref{bm regression constr} is equivalent to (see Corollary \ref{cor mean-abs dev})
\begin{align}
    \min_{f \in \cC} &\quad \|Z_f \|_1 \label{l1 obj}\\
    \textrm{s.t.} &\quad \bbe[Z_f] = 0 \label{l1 constr},
\end{align}
which is the unbiased  $\cL^1$ regression problem, hereafter referred to as the superexpectation (SE) regression. 

According to Theorem \ref{Th: equivalence of qr and bmr} and \cite[Regression Theorem]{Quadrangle}, the optimal solution to \eqref{l1 obj}--\eqref{l1 constr} is $f^*(\mathbf{x}) = \var_\alpha(Y|\mathbf{X} = \mathbf{x}),$ where $\alpha \in (0,1)$ is such that $0 = \bbe[Z_{f^*}] = \var_\alpha(Z_{f^*}),$ implying that
\begin{equation}\label{eq:unbiased quantile}
    \bbe[\var_\alpha(Y|\mathbf{X})] = \bbe[Y],
\end{equation}
when the conditional distribution of $Y|\mathbf{X} = \mathbf{x}$ is continuous.
\end{Cor}
It is well-known that the unbiased $\cL^2$ regression is simply the $\cL^2$ norm minimization
\begin{equation}\label{l2 regression problem}
    \min_{f \in \cC} \quad \|Z_f\|_2,
\end{equation}
since the optimal solution to \eqref{l2 regression problem} is the conditional mean\footnote{assuming that $f^* \in \cC$} $f^*(\mathbf{x}) = \bbe[Y|\mathbf{X} = \mathbf{x}]$ or more generally $f^*(\mathbf{X}) = \bbe[Y|\mathbf{X}],$ which is unbiased due to the ``tower property'' of the conditional expectation, i.e., $\bbe[\bbe[Y|\mathbf{X}]] = \bbe[Y].$

Given \eqref{eq:unbiased quantile}, the natural question is when $\bbe[Y|\mathbf{X}] = \var_\alpha(Y|\mathbf{X})?$ It is possible to answer this question in the case of linear regression, i.e., when $$\cC = \operatorname{Lin}(\bbr^{d}):= \lset f:\bbr^{d}\to\bbr \mset f(\mathbf{x}) = c_0 + \mathbf{c}^\top\mathbf{x} \quad \forall (c_0,\mathbf{c}) \in \bbr^{d+1} \rset$$ is a class of linear functions on $\bbr^{d}.$

\subsection{Linear Regression and Estimation}
Estimation via linear regression is one of the central topics of the RQ theory and classical statistics. It is well-known (see \cite[Regression Theorem]{Quadrangle}) that given a functional of expectation type, i.e. $\cE(X) = \bbe[e(X)],$ where $e: \bbr \to \bbr$ is such that $\cE$ is a (sub)regular error, the solution to the generalized regression problem \eqref{regression problem} is a conditional statistic $\cS(Y|\mathbf{X})$ related to $\cE$ (i.e., $\cE$ and $\cS$ are in the same quadrangle). The natural question is whether this is the case for $\cE_0(X) =  \max \{\bbe[X_+], \bbe[X_-] \}$ for which the associated statistic is $\bbe[X]$. Here, we answer this question in the case of linear regression.

Before stating the main theorem, we provide one technical definition needed for its proof.
\begin{Def}[Expectation of  random set]
    For a measurable set-valued mapping $S$, we say that $S(X)$ is a \emph{random set,} with possible outcomes $\{S(X), \ X\in \cX \ (\textrm{support of} \ X)\}$ occurring according to the probability distribution of $X.$ Its \emph{expectation} is defined as $$\bbe[S(X)]= \left\{ \bbe[v(X)] \mset v(X) \in S(X) \ \forall \; X \in \cX, \ v(X) \text{integrable}\right\}$$ see \cite{RoysetWets2021} for more details. 
\end{Def}

\begin{Th}[Equivalence of OLS and SE Regressions]\label{th: equiv of ols and bmr}
Let $X_i, \ i = 1,\ldots, d, \ Y$ and, $\varepsilon$ be random variables defined on $\cL^1(\Omega),$ and let 
\begin{equation}\label{true model}
    Y = f^*(\mathbf{X}) + \varepsilon \quad a.s., \quad f^* \in \operatorname{Lin}(\bbr^{d}), \quad \mathbf{X} = (X_1, \ldots, X_d)^\top
\end{equation}
be a true regression law, where (a) $\varepsilon$ and $X_i$ are independent, (b) $\bbe[\varepsilon] = 0 \, .$
Then 
\begin{equation}\label{eq: ols and se}
    (i) \ f^* =  \argmin\limits_{f \in \operatorname{Lin}(\bbr^d)} \|Z_f\|_2  \ \textrm{ if }  \ \|Z_{f^*}\|_2 < \infty, \qquad (ii) \ f^* \in  \argmin\limits_{f \in \operatorname{Lin}(\bbr^d)} \cE_0(Z_f) \, .
\end{equation}    
\end{Th}
\begin{proof}
 \underline{\emph{Proof of (i)}}.
The proof directly follows from the  \cite[Regression Theorem]{Quadrangle}.

\underline{\emph{Proof of (ii)}}. Let us consider the linear regression problem 
\begin{equation*}
  \min\limits_{f \in \operatorname{Lin}(\bbr^d)} \cE_0(Z_f) \, , 
\end{equation*}
which can be written in explicit form as follows
\begin{equation}\label{pw lin regr}
\min_{c_0, \mathbf{c}} \max \{\bbe[Z(c_0, \mathbf{c})_+],  \bbe[Z(c_0, \mathbf{c})_-]\} \,    
\end{equation}
where $Z(c_0, \mathbf{c}):= Y - c_0 - \mathbf{c}^\top \mathbf{X}.$ Problem \eqref{pw lin regr} is a convex optimization problem. An optimal solution $(c^*_0, \mathbf{c}^*)$ of \eqref{pw lin regr} satisfies the following necessary and sufficient first-order optimality condition
\begin{equation}\label{pw extr cond cont}
    (0, \mathbf{0}) \in \partial \bigl (\max \{\bbe[Z(c^*_0, \mathbf{c}^*)_+],  \bbe[Z(c^*_0, \mathbf{c}^*)_-]\}\bigr),
\end{equation}
where $\partial (\cdot)$ denotes the subdifferential of a convex function. By the Dubovitsky-Milyutin theorem, the condition \eqref{pw extr cond cont} is equivalent to the following system
\begin{equation}\label{subdif dub-mil cont}
\left\{\begin{matrix}
 (0,\mathbf{0}) \in \textrm{conv}\left(\partial\left(\bbe[Z(c^*_0, \mathbf{c}^*)_+]\right), \partial\left(\bbe[Z(c^*_0, \mathbf{c}^*)_-]\right)\right),\\
 \bbe[Z(c^*_0, \mathbf{c}^*)_+]=\bbe[Z(c^*_0, \mathbf{c}^*)_-], 
\end{matrix}\right.   
\end{equation}
where $\textrm{conv}(\cdot \, , \cdot)$ denotes the convex hull of the union of two sets.

Let us compute the subdifferentials $\partial\left(\bbe[Z(c^*_0, \mathbf{c}^*)_+]\right)$ and $\partial\left(\bbe[Z(c^*_0, \mathbf{c}^*)_-]\right)$ in \eqref{subdif dub-mil cont}. Since $x \mapsto x_{\pm}$ is a proper, convex, and lower-semicontinuous function, the subdifferentiation and the expectation operators can be interchanged. Thus
\begin{align}
 \partial\bbe[Z(c^*_0, \mathbf{c}^*)_+] &= \bbe\left[\begin{cases}
-(1,\mathbf{X}), & Z(c^*_0, \mathbf{c}^*)>0 \\
(0,\mathbf{0}), & Z(c^*_0, \mathbf{c}^*)< 0 \\
\textrm{conv}\left(-(1,\mathbf{X}),(0,\mathbf{0})\right), & Z(c^*_0, \mathbf{c}^*) = 0 
\end{cases}  \right],  \label{dub-mil for max plus} \\ 
\partial\bbe[Z(c^*_0, \mathbf{c}^*)_-] &= \bbe\left[\begin{cases}
(1,\mathbf{X}), & Z(c^*_0, \mathbf{c}^*)<0 \\
(0,\mathbf{0}), & Z(c^*_0, \mathbf{c}^*)> 0 \\
\textrm{conv}\left((1,\mathbf{X}),(0,\mathbf{0})\right), & Z(c^*_0, \mathbf{c}^*) = 0
\end{cases}  \right]. \label{dub-mil for max minus}
\end{align}

Let $f^*(\mathbf{X}) = a_0 + \mathbf{a}^\top \mathbf{X}$. To prove the theorem, we will show that $(a_0, \mathbf{a})$ satisfies conditions \eqref{subdif dub-mil cont}. Let $c_0^* = a_0$ and $\mathbf{c}^* = \mathbf{a}$, then \eqref{true model} implies that $Z(a_0, \mathbf{a}) = \varepsilon.$ Hence, given \eqref{dub-mil for max plus} and \eqref{dub-mil for max minus}, system \eqref{subdif dub-mil cont} is equivalent to 
\begin{equation}\label{extr condition 1}
 \left\{\begin{matrix}
(0,\mathbf{0}) \in \textrm{conv}\left(\bbe\left[\begin{cases}
-(1,\mathbf{X}), &\varepsilon>0 \\
(0,\mathbf{0}), & \varepsilon< 0 \\
\textrm{conv}\left(-(1,\mathbf{X}),(0,\mathbf{0})\right), & \varepsilon = 0 
\end{cases}  \right], \bbe\left[\begin{cases}
(1,\mathbf{X}), & \varepsilon<0 \\
(0,\mathbf{0}), & \varepsilon> 0 \\
\textrm{conv}\left((1,\mathbf{X}),(0,\mathbf{0})\right), & \varepsilon = 0
\end{cases}  \right] \right) \\
\\
 \bbe[\varepsilon_+]=\bbe[\varepsilon_-]
\end{matrix} \right . 
\end{equation}
  Denote by $\mathfrak{M}_- := \bbe\big[ \mathfrak{C}_-(\mathbf{X}) \mset \varepsilon=0\big]$ the conditional expectation of a random set $\mathfrak{C}_-(\mathbf{X}):=\textrm{conv}\left(-(1,\mathbf{X}),(0,\mathbf{0})\right) = \{-\lambda(1,\mathbf{X})\mset \lambda \in [0,1]\},$ by $\mathfrak{M}_+ := \bbe\big[ \mathfrak{C}_+(\mathbf{X}) \mset \varepsilon=0\big]$ the conditional expectation of a random set $\mathfrak{C}_+(\mathbf{X}):=\textrm{conv}\left((1,\mathbf{X}),(0,\mathbf{0})\right) = \{\lambda(1,\mathbf{X})\mset \lambda \in [0,1]\},$ and by $\varkappa = \bbp(\varepsilon=0).$ Define 
 $$\mathfrak{H}(\varepsilon):= \operatorname{conv} \left(-\nu_+\big( 1, \bbe\big[\mathbf{X}|\varepsilon > 0\big]\big) + \varkappa \mathfrak{M}_-,  \ \nu_-\big( 1, \bbe\big[\mathbf{X}|\varepsilon < 0\big]\big) + \varkappa \mathfrak{M}_+\right).
 $$
 Then the system \eqref{extr condition 1} can be rewritten in terms of conditional expectation
\begin{equation}\label{opt cond linear err 2}
 \left\{\begin{matrix}
(0, \mathbf{0}) \in \mathfrak{H}(\varepsilon)\\
\bbe[\varepsilon_+]=\bbe[\varepsilon_-]
\end{matrix}\right. 
\end{equation}
where $\nu_+ = \bbp(\varepsilon>0)$ and $\nu_- = \bbp(\varepsilon < 0)$. Assumption \emph{(a)} and the fact that $0 \in \bbe[\mathfrak{C}_+(\mathbf{X})]$ and $0 \in \bbe[\mathfrak{C}_-(\mathbf{X})]$ imply that the first condition in \eqref{opt cond linear err 2} is satisfied since 
$$(0, \mathbf{0}) \in \textrm{conv} \left(-\nu_+\big( 1, \bbe\big[\mathbf{X}\big]\big) + \varkappa \bbe[\mathfrak{C}_-(\mathbf{X})],  \ \nu_-\big( 1, \bbe\big[\mathbf{X}\big]\big) + \varkappa \bbe[\mathfrak{C}_+(\mathbf{X})]\right) \, .
$$
 Finally, according to the assumption \emph{(b)}, the second condition in \eqref{opt cond linear err 2} is also satisfied due to the Lebesgue decomposition of $\varepsilon$, i.e.,   $\varepsilon = \varepsilon_+ - \varepsilon_-$. 
 Therefore,  $(a_0, \mathbf{a}) \in \min\limits_{c_0, \mathbf{c}} \cE_0(Z(c_0, \mathbf{c}))$, which completes the proof.
\end{proof}

\begin{Cor}[Interpretation via quantile regression]\label{cor:three regressions}
Results obtained in Theorem \ref{Th: equivalence of qr and bmr}  and Theorem \ref{th: equiv of ols and bmr} imply that the SE regression estimates the conditional mean by automatically estimating the conditional quantile with confidence level $\alpha$ such that 
\begin{equation} \label{eq: residual}
   0 = \bbe[Z_{f^*}] \in  \var_\alpha(Z_{f^*}), \quad f^* \in  \argmin\limits_{f \in \operatorname{Lin}(\bbr^d)} \cE_0(Z_f) 
\end{equation}
In particular, when the joint distribution of the regressant $Y$ and regressor $\mathbf{X}$ is continuous, inclusions in \eqref{eq: residual} can be replaced with equality.

Moreover, let $\mathfrak{S}, \mathfrak{S}_0, \mathfrak{S}_\alpha$ denote a set of optimal solutions to the linear regression problems with $\cL^2$, $\cE_0$, and $\cE_\alpha$ errors, respectively. Then the following chain of inclusions holds:
\begin{equation}\label{eq: chain}
    \mathfrak{S} \subseteq \mathfrak{S}_0 \subseteq \mathfrak{S}_\alpha \, 
\end{equation}
with alpha defined in \eqref{eq: residual}. 

It is worth noting that similar results were obtained by \cite{Anton2024expectile} for expectile regression. 
\end{Cor}
% \begin{rem}[Relation to expectile regression]
    
% \end{rem}

\subsection{Sample Average Approximation and Optimization}
Exact solving \eqref{eq: ols and se} requires access to the joint distribution of $(Y, \mathbf{X})$, which is unknown in most practical situations. Instead, one works with an i.i.d. random sample $\{(y_i, \mathbf{x}_i)\}_{i=1}^n$ assumed to be drawn from that joint distribution. A standard approach in this setting is the Sample Average Approximation (SAA) method (see, e.g., \cite{stochproglec}), which replaces expectations by empirical averages and solves the corresponding deterministic optimization problem.

The application of SAA to ordinary least squares (OLS) regression (see item (i) in \eqref{eq: ols and se}) is well studied and well understood in the classical statistics literature. By contrast, SE regression (see item (ii) in \eqref{eq: ols and se}) is a compound stochastic programming problem. The application of SAA to this type of problem has been studied by \cite{SAAErmoliev}. Moreover, given the equivalent formulation of SE regression as the unbiased $\cL^1$ regression problem, classical tools--such as uniform (or epigraphical) laws of large numbers, functional central limit theorems, and functional large-deviation results--are also applicable \citep{ArtsteinWets1995Consistency, stochproglec}.

The SAA formulation of the SE regression is as follows
\begin{equation}\label{SE regr SAA}
 \min_{c_0, \mathbf{c}} \max \left\{\frac{1}{n}\sum_{i=1}^n z_i(c_0, \mathbf{c})_+,  \frac{1}{n}\sum_{i=1}^n z_i(c_0, \mathbf{c})_- \right\} \,  ,  
\end{equation}
where $z_i(c_0, \mathbf{c}) = y_i - c_0 - \mathbf{c}^\top \mathbf{x}_i, \ i = 1,\ldots, n$. Formulation \eqref{SE regr SAA} is a piece-wise linear convex optimization problem, which can be efficiently solved by methods of nondifferentiable (a.k.a. nonsmooth) optimization \citep{Shor1998Nondifferentiable}. This property of \eqref{SE regr SAA} does not provide an algorithmic edge over OLS (smooth problems are usually solved faster than nonsmooth ones \citep{nesterov2018lectures}), yet the piecewise-linear objective yields statistical and modeling benefits.

\textit{Robustness to heavy tails and outliers.}
The objective in \eqref{SE regr SAA} is convex and piecewise linear. Consequently, it requires merely a finite first moment for well-posedness and consistency, in contrast to quadratic losses that demand finite second moments and exhibit pronounced sensitivity to extreme observations. Practically, this translates into stable estimation under heavy-tailed errors and outliers, without the tuning required by Huber-type losses \citep{Huber1964, HuberRonchetti2009} and Support Vector Regression method \citep{SVR, SVR2, NewSVM, SVRTutorial, Anton2022SVR}.

\textit{Formulation as linear programming.}
Problem \eqref{SE regr SAA} can be equivalently cast as a linear programming (LP) problem:
\begin{equation}\label{SE regr SAA LP}
  \begin{split}
            \min_{c_0, \mathbf{c},t, \mathbf{u}} &\quad t\\
           \textrm{s.t.} &\quad t \geq\frac{1}{n}\sum_{i=1}^n u_i \, ,\\
            &\quad  t \geq \frac{1}{n}\sum_{i=1}^n u_i - \frac{1}{n}\sum_{i=1}^n z_i(c_0,\mathbf{c}) \, ,\\
            &\quad u_i \geq 0 \, ,    \ \qquad\qquad i = 1,\ldots, n\\
             &\quad u_i \geq z_i(c_0,\mathbf{c}),  \   \ \  \quad  i = 1,\ldots, n \, .
        \end{split}   
\end{equation}
Formulation \eqref{SE regr SAA LP} has several modeling and optimization advantages. Specifically, it naturally supports linear constraints, effectively handles ill-conditioning due to the maturity of modern LP solvers (e.g., \cite{gurobi}), and more importantly, allows for efficient incorporation of integer variables, which is the subject of Section \ref{sec: Sparse Regression}. Additionally, recent advances in numerical algorithms for LP \citep{ApplegateEtAl2021PDLP} enable GPU acceleration, as implemented by \cite{NVIDIA_cuOpt_UserGuide_25_08}, achieving over 5,000x faster performance compared to traditional CPU-based solvers. These methods apply to both continuous LPs and mixed-integer LPs (MILPs).

\subsection{Sparse Regression} \label{sec: Sparse Regression}
Sparse regression (or best-subset selection) seeks generalized linear models by encouraging most coefficients to be exactly zero, thereby improving interpretability and prediction in high-dimensional settings where the number of regressors, $d$, may exceed the number of observations, $n$. Classical best-subset selection minimizes squared error with an $\ell_0$ pseudo-norm penalty, leading to NP-hard mixed-integer quadratic programming: 
\begin{equation}\label{OLS sparse}
   \min_{c_0, \mathbf{c}} \  \frac{1}{n}\sum_{i=1}^n z_i^2(c_0, \mathbf{c}) \quad \textrm{s.t.} \quad  \|\mathbf{c}\|_0 \leq k \, , 
\end{equation}
where $\|\mathbf{c}\|_0: = \operatorname{card}(\{j: c_j \neq 0\})$ and $d \geq k \in \mathbb{N}$.

Convex relaxations such as the Lasso replace $\ell_0$ with $\ell_1$ norm to yield tractable estimators with provable prediction and support-recovery guarantees under suitable design conditions (e.g., restricted eigenvalue/compatibility). However, due to the recent progress in the computational power of Mixed Integer Programming (MIP) solvers (see e.g., \cite{Bixby2012BriefHistory}), the original $\ell_0$ penalty formulation became more and more widely used. Moreover, MIP formulations of sparse regression \eqref{OLS sparse} have been shown to deliver sharper support-recovery performance than popular convex relaxations (like the Lasso) under comparable design and signal-strength conditions; see phase-transition and empirical results in \citet{BertsimasKingMazumder2016BestSubset,BertsimasVanParys2020SparseHD,BertsimasPauphiletVanParys2020SparseRegreview}.

There are two most common approaches to solving \eqref{OLS sparse}: 1) direct mixed-integer quadratic programming (MIQP) with branch-and-bound, cuts, and heuristics features efficiently implemented in commercial solvers, e.g., \cite{gurobi}; 2) cutting-plane (outer-approximation) scheme that iteratively adds valid inequalities to tighten the relaxation, leading to solving a sequence of MILPs \citep{BertsimasVanParys2020SparseHD}. Alternatively, given the results obtained in the sections above, we propose to solve  
\begin{equation}\label{SE regr SAA sparse}
 \min_{c_0, \mathbf{c}} \max \left\{\frac{1}{n}\sum_{i=1}^n z_i(c_0, \mathbf{c})_+,  \frac{1}{n}\sum_{i=1}^n z_i(c_0, \mathbf{c})_- \right\} \quad \textrm{s.t.} \quad \|\mathbf{c}\|_0 \leq k \, ,  
\end{equation}
which can be equivalently formulated as an MILP in view of \eqref{SE regr SAA LP}. Compared with the MIQP best-subset formulation \eqref{OLS sparse}, the SE regression model in \eqref{SE regr SAA sparse} yields a pure MILP, so every branch-and-bound node solves a linear program rather than a quadratic one--typically cheaper, more scalable, and numerically stabler. MILPs also unlock a larger, more mature ecosystem of cutting planes and presolve techniques, and make it straightforward to add linear side-constraints (e.g., monotonicity, fairness, etc.) without changing problem class. In practice, MILP nodes can leverage first-order LP technology (e.g., PDLP of \cite{ApplegateEtAl2021PDLP} mentioned earlier) and even GPU acceleration, which is not available for generic MIQPs. Finally, the piecewise-linear objective often leads to tighter and more informative dual bounds for tail- or imbalance-aware formulations, while avoiding conditioning issues common in quadratic losses.

\section{Case Studies}\label{sec: Case Studies}
This section conducts four case studies:
\begin{itemize}
  \item Equivalence of CVaR and superexpectation risk portfolio optimization;
  \item Equivalence between quantile regression and biased-mean regression;
  \item Equivalence between OLS and superexpectation regression;
  \item Sparse regression: OLS vs.\ superexpectation error.
\end{itemize}
\paragraph{Hardware:} Intel(R) Core(TM) i5-11400H @ 2.70GHz (12 CPUs).
\paragraph{Software:} Python \citep{python-psf}, Portfolio Safeguard \citep{PSG}, Gurobi \citep{gurobi}.

% \paragraph{Code:} I will include the link

\subsection{Equivalence of CVaR and Superexpectation Risk Portfolio Optimization}\label{subsec:portfolio_equival}
This subsection numerically validates the equivalence of portfolio optimization (see Example~\ref{portfolio opt example}) with SE deviation and CVaR deviation objectives. We consider the data set with $4$ assets and $10,000$ observations from the case study \cite{POEX}. 
% Codes, data, solutions, and mathematical problem formulations are posted at the link \#(see Problem \#). 

First, we solved the SE deviation portfolio optimization problem $25$ times by varying the parameter $x$ from $-10^{-4}$ to $5\times10^{-2}$ with a step $0.0020875$ and calculated the corresponding CVaR deviation parameter $\alpha$ using Proposition~\ref{prop:equivalence}. Second, we computed the CVaR deviation objective at the optimal SE deviation points for corresponding $\alpha$ and $x$. Third, we solved the CVaR deviation portfolio optimization problem $25$ times for the $\alpha$'s obtained in the first step and calculated the SE deviation objective at the optimal CVaR points for corresponding values of $x$. 

Figure~\ref{fig:portfolio} compares the optimal objectives of CVaR deviation and SE deviation portfolio optimization problems versus the objectives computed at optimal points for the SE deviation and CVaR deviation, respectively. We observe that the optimal SE deviation and CVaR deviation portfolios coincide with high precision for all corresponding $x$ and $\alpha$.

\begin{figure}[h!]
    \centering
    \includegraphics{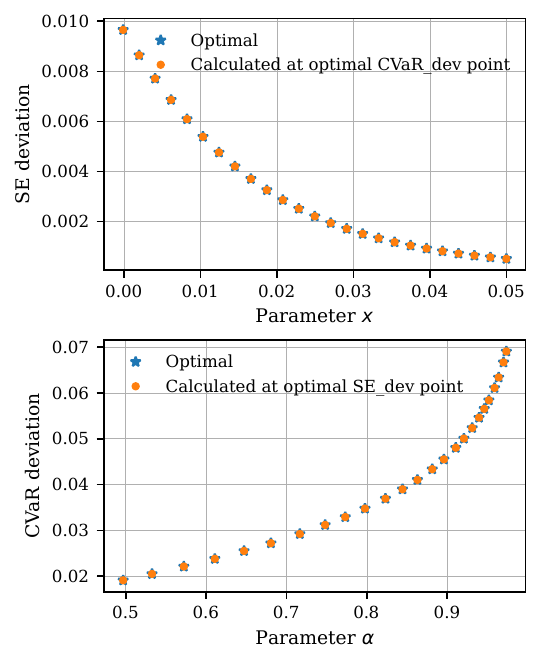}
    \caption{Portfolio optimization with SE deviation and CVaR deviation objectives.}
    \label{fig:portfolio}
\end{figure}

\subsection{Equivalence between Quantile Regression and Biased Mean Regression}
This subsection numerically confirms the equivalence between quantile regression and biased mean regression (see Theorem~\ref{Th: equivalence of qr and bmr}). We use the style classification data set containing $1,264$ observations and $4$ factors (regressors), previously considered in the case study \cite{SCQR}.

To compare estimates obtained with different error functions, we solved the problem \eqref{regression problem} with the SE error and the KB error for the same data set. We considered the case $x=0.005$ and $\alpha^* = F_{Z(c_0^*,\mathbf{c}^*)}(0),$ where $(c_0^*,\mathbf{c}^*)$ is an optimal solution to the regression problem with the SE error. Table \ref{tab:2} compares the results of calculations for these two error functions, where
Row 1 = parameters of the error functions.
Rows 2--7 = optimal values of explanatory variables and intercept. Row 8 =  value of error functions at the optimal point for the SE error. Row 9 =  value of error functions at the optimal point of the KB error.
\begin{table}[H]
    \centering
    \begin{tabular}{|c|c||c||c|}
        \hline
        \# & Objective & SE Error & KB Error \\ \hline
         1& Parameter value& $x=0.005$& $\alpha = 0.803$\\ \hline
         2& Optimal point:& &  \\ \hline
         3&\textsc{rlv} &0.551130 & 0.551001  \\ \hline
         4& \textsc{rlg}& 0.506271 & 0.506568\\ \hline
         5& \textsc{ruj}& -0.072197& -0.072300 \\ \hline
         6& \textsc{ruo}&-0.005550 & -0.005552   \\ \hline
         7&intercept &0.004083 &0.004088   \\ \hline
         8&Errors at optimal point of \eqref{bm regression problem}&0.000789  &0.001774  \\ \hline
         9&Errors at optimal point of \eqref{q regression problem}&0.000789  &0.001774   \\ \hline
    \end{tabular}
    \caption{Biased mean estimation via regressions with SE error and KB error. Row 1 = parameters of the error functions.
Rows 2--7 = optimal values of explanatory variables and intercept. Row 8 =  value of error functions at the optimal point for the SE error. Row 9 =  value of error functions at the optimal point of the KB error. }
    \label{tab:2}
\end{table}

The results in Table \ref{tab:2} numerically validate  Theorem~\ref{Th: equivalence of qr and bmr}. We observe that the coefficients obtained by the two regressions have identical values with at least 3-digit precision. The objective value of the KB error computed at the optimal point of the SE error (row 8, column 4) coincides with the true optimal objective value (row 9, column 4) with 6-digit precision. The same holds for the objective value of the SE error computed at the optimal point for the corresponding KB error (row 8, column 3 and row 9, column 3).

\subsection{ Equivalence between OLS and Superexpectation Regression}
This subsection conducts a controlled experiment, which numerically confirms Theorem~\ref{th: equiv of ols and bmr} (and Corollary~\ref{cor:three regressions}).
% Calculation results are posted at the web link \cite{ECX} (see Problems 5a, 5b, 5c).
As a true law, we take
\begin{equation}\label{simulation}
    Y = X + \varepsilon,
\end{equation}
where $X \sim \mathcal{N}(0,1)$ and $\varepsilon \sim \mathcal{SN}(10) $, where $ \mathcal{SN}(10)$ is the skew normal distribution (see \cite{Azzalini1985}) with shape parameter $a=10$ and location and scale parameters such that the mean and standard deviation are zero and one, respectively.
We estimate the regression coefficients with the SE error ($x=0$), the mean square error (MSE), and the KB error ($\alpha^* = F_\varepsilon(0) = 0.572760$). With given $\varepsilon,$ the assumptions \emph{(a),(b)} of the Theorem \ref{th: equiv of ols and bmr} are satisfied.

We generated $100$ samples of size $n = 100, \ 500, \  1000, \ 5000,  \ 10000, \ 50000, \ 100000, 500000$ with Monte-Carlo method for \eqref{simulation}. For each $n$, we solved $100$ regression problems with three errors and 
computed:
\begin{enumerate}
  \item Minimal Relative $\ell^2$ Error $= \displaystyle{\min_{1 \leq j \leq 100} }\dfrac{\norm{\mathbf{c}^{*,k}_j - \mathbf{c}^*}_2}{\norm{\mathbf{c}^{*,k}_j}_2}, \quad k = 1,2,3$;
    \item Average Relative $\ell^2$ Error $ = \displaystyle{\frac{1}{100}\sum\limits_{j = 1}^{100}}\dfrac{\norm{\mathbf{c}^{*,k}_j - \mathbf{c}^*}_2}{\norm{\mathbf{c}^{*,k}_j}_2}, \quad k = 1,2,3$;  
     \item Maximal Relative $\ell^2$ Error $= \displaystyle{\max_{1 \leq j \leq 100} }\dfrac{\norm{\mathbf{c}^{*,k}_j - \mathbf{c}^*}_2}{\norm{\mathbf{c}^{*,k}_j}_2}, \quad k = 1,2,3$;

    \item  Spread  for Relative $\ell^2$ Error     
    $= \displaystyle{\max_{1 \leq j \leq 100} }\dfrac{\norm{\mathbf{c}^{*,k}_j - \mathbf{c}^*}_2}{\norm{\mathbf{c}^{*,k}_j}_2} - \displaystyle{\min_{1 \leq j \leq 100} }\dfrac{\norm{\mathbf{c}^{*,k}_j - \mathbf{c}^*}_2}{\norm{\mathbf{c}^{*,k}_j}_2}, \quad k = 1,2,3$,
\end{enumerate}
where $\mathbf{c}^{*,1}_j = (c_0^{*,1},c_1^{*,1})^\top_j$ is an optimal solution to the regression \eqref{regression problem} with  the SE error for the $j$th sample, $\mathbf{c}^{*,2}_j = (c_0^{*,2},c_1^{*,2})^\top_j$ is an optimal solution to the regression \eqref{regression problem} with  the mean square error for the $j$th sample, and $\mathbf{c}^{*,3}_j = (c_0^{*,3},c_1^{*,3})^\top_j$ is an optimal solution to the regression \eqref{regression problem} with  the KB error for the $j$th sample.

Table~\ref{tab3}, Table~\ref{tab4}, and Table~\ref{tab:5} show that as the sample size increases, the average relative  $\ell^2$ error monotonically decreases. Solution vectors for three regression problems converge to the true solution $\mathbf{c}^* = (0,1)^\top.$  The numerical convergence rate is similar for all regressions.
\begin{table}[H] % placement parameter H
\centering
\begin{tabular}{|c||c||c||c||c|}
\hline
Sample Size& $\;\;$Minimal$\;\;$ & $\;\;$Average$\;\;$ &$\;$ Maximal$\;$ & $\;\;$Spread $\;\;$ \\
\hline
 100 & 0.016747 & 0.134246 & 0.315741 & 0.298994 \\
 \hline
 500 & 0.006054 & 0.058069 & 0.152820 & 0.146766 \\
 \hline
 1000 & 0.007891 & 0.043216 & 0.130886 & 0.122995 \\
 \hline
 5000 & 0.003236 & 0.017786 & 0.051094 & 0.047857 \\
 \hline
 10000 & 0.000815 & 0.011717 & 0.044784 & 0.043969 \\
 \hline
 50000 & 0.000343 & 0.005583 & 0.012998 & 0.012655 \\
 \hline
 100000 & 0.000406 & 0.003948 & 0.010341 & 0.009935 \\
 \hline
 500000 & 0.000126 & 0.001823 & 0.005316 & 0.005190 \\
\hline
\end{tabular}
\caption{Numerical convergence of regression coefficients with MSE. 
 Relative Minimal, Average, Maximal, Spread for $\ell^2$ error. }
\label{tab3}
\end{table}

\begin{table}[H]
\centering
\begin{tabular}{|c||c||c||c||c|}
\hline
Sample Size& $\;\;$Minimal$\;\;$ & $\;\;$Average$\;\;$ &$\;$ Maximal$\;$ & $\;\;$Spread $\;\;$ \\
\hline
 100 & 0.021023 & 0.161826 & 0.387331 & 0.366308 \\
 \hline
 500 & 0.003810 & 0.072019 & 0.185431 & 0.181621 \\
 \hline
 1000 & 0.004552 & 0.049969 & 0.181418 & 0.176866 \\
 \hline
 5000 & 0.003213 & 0.021631 & 0.064678 & 0.061464 \\
 \hline
 10000 & 0.000540 & 0.015914 & 0.053236 & 0.052696 \\
 \hline
 50000 & 0.000649 & 0.007127 & 0.020595 & 0.019946 \\
 \hline
 100000 & 0.000752 & 0.004364 & 0.009795 & 0.009043 \\
 \hline
 500000 & 0.000121 & 0.002197 & 0.006262 & 0.006141 \\
\hline
\end{tabular}
\caption{Numerical convergence of regression coefficients with SE error. Relative Minimal, Average, Maximal, Spread for $\ell^2$ error. }
\label{tab4}
\end{table}

\begin{table}[H]
\centering
\begin{tabular}{|c||c||c||c||c|}
\hline
Sample Size& $\;\;$Minimal$\;\;$ & $\;\;$Average$\;\;$ &$\;$ Maximal$\;$ & $\;\;$Spread $\;\;$ \\
\hline
 100 & 0.029656 & 0.187235 & 0.431448 & 0.401792 \\
 \hline
 500 & 0.003415 & 0.082709 & 0.203619 & 0.200204 \\
 \hline
 1000 & 0.001789 & 0.059149 & 0.182991 & 0.181201 \\
 \hline
 5000 & 0.002448 & 0.024226 & 0.064841 & 0.062393 \\
 \hline
 10000 & 0.003892 & 0.018588 & 0.057114 & 0.053222 \\
 \hline
 50000 & 0.000684 & 0.007947 & 0.019892 & 0.019208 \\
 \hline
 100000 & 0.000131 & 0.005065 & 0.011583 & 0.011452 \\
 \hline
 500000 & 0.000263 & 0.002580 & 0.005907 & 0.005643 \\
\hline
\end{tabular}
\caption{Numerical convergence of regression coefficients with KB error. Relative Minimal, Average, Maximal, Spread for $\ell^2$ error.}
 \label{tab:5}
\end{table}

\subsection{Sparse Regression: Mean Square Error vs.\ Superexpectation Error}\label{subsec:sparse_numerical}
This subsection numerically compares the performance of the sparse regression with MSE \eqref{OLS sparse} and SE error \eqref{SE regr SAA sparse}. We are considering a simple setting where the true model is 
\begin{equation}\label{sparse simulation}
    Y = \mathbf{c}^{*\top}\mathbf{X} + \varepsilon \, ,
\end{equation}
where $\mathbf{X} \sim \mathcal{N}(\mathbf{0},\Sigma)$ with $\Sigma = (\rho^{|i-j|})_{i,j}, \ \rho = 0.9$, $\varepsilon \sim \mathcal{N}(0,1)$, and $\mathbf{c}^* \in \{-1,0,1\}, \ \mathbf{c}^* \in \bbr^{d}, \ d=3000$ with exactly $k^*=10$ nonzero coefficients. 

We generated $10$ samples of size $n = 300, \ 500, \  1000, \ 5000$ with Monte-Carlo method for \eqref{sparse simulation}. For each $n$, we solved $10$ sparse regression problems with two errors and 
computed:
\begin{enumerate}
    \item Minimal Accuracy $= \min\limits_{1 \leq j \leq 10} A(\mathbf{c}^m_j), \quad A(\mathbf{c}^m_j):=\dfrac{\operatorname{card}(\{i: c_{i,j}^m \neq 0, c_i^* \neq 0\ \})}{k^*}, \quad  j = 1,\ldots,n;$
    \item Maximal Accuracy $=\max\limits_{1 \leq j \leq 10} A(\mathbf{c}^m_j), \quad m=1,2$
    \item Average Accuracy $= \dfrac{1}{10}\displaystyle{\sum\limits_{j=1}^{10}}  \, A(\mathbf{c}^m_j), \quad m=1,2,$
\end{enumerate}
where $\mathbf{c}_j^m = (c_{1,j}^m, \ldots, c_{d,j}^m)$ is an optimal (or at least feasible) solution to the sparse regression problem with MSE ($m=1$), and with SE error ($m=2$) obtained in $300$ seconds for the $j$'s sample.

Tables~\ref{tab6} and~\ref{tab7} present the numerical performance of the two regressions in terms of accuracy, solution time, and MIP gap. The results indicate that the sparse regression with SE error substantially outperforms the MSE-based formulation when the number of observations $n$ is much smaller than the number of regressors ($d=3000$). By contrast, as $n$ increases, Gurobi solves the MSE formulation to proven optimality much faster; for $n=5000$, optimality is certified in approximately $200$ seconds. From a practical standpoint, regimes with $n \ll d$ are often more relevant, and in that regime the SE error is particularly effective.
  
\begin{table}[H] % placement parameter H
\centering
\begin{tabular}{|c||c||c||c||c||c|}
\hline
Sample Size& $\;\;$Min Acc.$\;\;$ & $\;\;$Avg Acc.$\;\;$ &$\;$ Max Acc.$\;$ &$\;$ Avg Time (sec)$\;$&$\;$ Avg MIP Gap$\;$\\
\hline
 300 & 0.9 & 0.99 & 1.0 & 300& 10.25 \%\\
 \hline
 500 & 1.0 & 1.0 & 1.0 & 300&9.68 \%\\
 \hline
 1000 & 1.0 & 1.0 & 1.0 &300 &7.02 \%\\
 \hline
 5000 & 0.0 & 0.0 & 0.0 & 300&69.84\%\\
 \hline
\end{tabular}
\caption{Numerical results for sparse regression with SE error. 
 Minimal, Average, Maximal Accuracy, average solving time, and MIP gap. }
\label{tab6}
\end{table}
\begin{table}[H] % placement parameter H
\centering
\begin{tabular}{|c||c||c||c||c||c|}
\hline
Sample Size& $\;\;$Min Acc.$\;\;$ & $\;\;$Avg Acc.$\;\;$ &$\;$ Max Acc.$\;$ &$\;$ Avg Time (sec)$\;$&$\;$ Avg MIP Gap$\;$\\
\hline
 300 & 0.0 & 0.0 & 0.0 & 300& --\\
 \hline
 500 & 0.0 & 0.1 & 1.0 & 300&--\\
 \hline
 1000 & 1.0 & 1.0 & 1.0 & 282 & 88 \%\\
 \hline
 5000 & 1.0 & 1.0 & 1.0 &204 &0 \%\\
 \hline
\end{tabular}
\caption{Numerical results for sparse regression with MSE. 
 Minimal, Average, Maximal Accuracy, average solving time, and MIP gap.}
\label{tab7}
\end{table}

\section{Conclusions}\label{sec:Conclusion}
This paper introduced the \emph{biased mean quadrangle} (BMQ) and developed an associated family of estimation and decision tools centered on a biased mean statistic. The BMQ’s risk---the \emph{superexpectation risk}---extends mean-upper-semirisk, while its error---the \emph{superexpectation error}---balances averages of the positive and negative parts of a random variable $Y$ in a tunable, directionally asymmetric manner. These choices yield models that directly target practically meaningful events (exceeding or falling short of the mean by a margin $x$ in the units of $Y$) and integrate with regret and deviation via the RQ identities (Definition~\ref{risk quadrangle}).

On the theoretical side, we established: (i) \emph{subregularity} of the BMQ (Definition~\ref{def:subregquadrangle}, Proposition~\ref{Regularity of BMQ}), thereby providing, to our knowledge, the first subregular quadrangle with concrete applied relevance; (ii) \emph{equivalence of deviation minimization} between CVaR and superexpectation deviations (Proposition~\ref{prop:equivalence}), with implications for portfolio optimization (Example~\ref{portfolio opt example}); and (iii) a \emph{reparameterization of quantile regression} in terms of biased means, proving that superexpectation-based biased-mean regression and quantile regression yield the same solutions for an appropriate mapping between the margin $x$ and the quantile level $\alpha$ (Theorem~\ref{Th: equivalence of qr and bmr}). As a special case, setting $x=0$ recovers a coherent, piecewise-linear mean quadrangle for which the induced \emph{superexpectation regression} is equivalent to OLS and admits an LP formulation (Corollary~\ref{cor:reg+coherence}, Theorem~\ref{th: equiv of ols and bmr}), offering modeling transparency and computational advantages.

From a practical standpoint, the unit-calibrated parameterization by $x$ enhances interpretability (``factors of beating the average exactly by $x$'') and alignment with domain policies. Our numerical experiments further indicate that, in small-sample high-dimensional regimes ($n \ll d$), the SE-based sparse regression is particularly effective, while for larger $n$ the MSE-based formulation is solved to proven optimality more rapidly by off-the-shelf solvers. Altogether, the BMQ provides a principled bridge between statistical estimation and risk-sensitive optimization.

\newpage
\bibliography{references}
\bibliographystyle{plainnat}
  \renewcommand{\bibsection}{\subsubsection*{References}}
\end{document}